\documentclass{eptcs}
\providecommand{\event}{GandALF 2023} 

\usepackage{iftex}

\ifpdf
  \usepackage{underscore}         
  \usepackage[T1]{fontenc}        
\else
  \usepackage{breakurl}           
\fi

\usepackage{latexsym}
\usepackage{color}
\usepackage{thmtools,thm-restate}
\usepackage{colonequals}
\usepackage{cancel}
\usepackage{tikz}
\usepackage{float}
\usepackage{paralist}
\usepackage{mathtools}

\usepackage{amsfonts, amsmath, amssymb, amsthm}

\newcommand{\eq}{\leftrightarrow}
\newcommand{\Eq}{\Leftrightarrow}
\newcommand{\imp}{\rightarrow}
\newcommand{\Imp}{\Rightarrow}

\newcommand{\Pmi}{\Leftarrow}
\newcommand{\et}{\wedge}

\renewcommand{\phi}{\varphi}

\newcommand{\inter}{\cap}
\newcommand{\Inter}{\bigcap}
\newcommand{\ce}{\colonequals}
\newcommand{\cce}{\coloncolonequals}

\newcommand{\M}{\widehat{K}}
\newcommand{\N}{\widehat{D}}
\newcommand{\localK}{\mathcal L _K^{\mathsf{loc}}}
\newcommand{\glocalK}{\mathcal L _K^{\mathsf{gloc}}}
\newcommand{\localD}{\mathcal L _D^{\mathsf{loc}}}
\newcommand{\glocalD}{\mathcal L _D^{\mathsf{gloc}}}
\newcommand{\fctdef}{\bowtie^{\mathcal F}}
\newcommand{\fctsat}{\Vvdash^{\mathcal F}}
\newcommand{\afdef}{\bowtie} 
\newcommand{\afsat}{\Vvdash} 
\newcommand{\nafsat}{\mathrel{\rlap{$\Vvdash$}\diagup}} 
\newcommand{\nfctsat}{\nafsat^{\mathcal F}}
\newcommand{\twoVsat}{\Vdash^{\mathcal F}}
\newcommand{\ntwoVsat}{\nVdash^{\mathcal F}}

\newcommand{\C}{\mathcal C}
\newcommand{\VV}{\mathcal V}
\newcommand{\FF}{\mathcal F}
\newcommand{\weg}[1]{}

\newcommand{\sstar}{\mathsf{star}}

\title{On Two- and Three-valued Semantics\\ 
for Impure Simplicial Complexes}

\author{Hans van Ditmarsch\institute{CNRS, Universit\'e de Toulouse, IRIT\\ Toulouse, France}
\email{hans.van-ditmarsch@irit.fr}
\and
Roman Kuznets \qquad\qquad Rojo  Randrianomentsoa
\institute{TU Wien
\qquad\qquad\strut
\\ 
Vienna, Austria
\qquad\qquad\strut}
\email{\{roman.kuznets+rojo.randrianomentsoa\}@tuwien.ac.at}\thanks{Funded by the Austrian Science Fund (FWF)  ByzDEL project (P33600).}
}
\def\titlerunning{On Two- and Three-valued Semantics for Impure Simplicial Complexes}
\def\authorrunning{H. van Ditmarsch, R. Kuznets \& R. Randrianomentsoa} 

\theoremstyle{plain}
\newtheorem{theorem}{Theorem}
\newtheorem{lemma}[theorem]{Lemma}
\newtheorem{proposition}[theorem]{Proposition}
\newtheorem{corollary}[theorem]{Corollary}

\theoremstyle{definition}
\newtheorem{definition}[theorem]{Definition}

\theoremstyle{remark}
\newtheorem{remark}[theorem]{Remark}
\newtheorem{example}[theorem]{Example}

\begin{document}

\maketitle

\begin{abstract}
Simplicial complexes are a convenient semantic primitive to reason about processes (agents) communicating with each other in synchronous and asynchronous computation. 
Impure simplicial complexes distinguish active processes from crashed ones, in other words, agents that are alive from agents that are dead. 
In order to rule out that dead agents reason about themselves and about other agents, three-valued epistemic semantics have been proposed where, in addition to the usual values true and false, the third value stands for undefined: 
the knowledge of dead agents is undefined and so are the propositional variables describing their local state. 
Other semantics for impure complexes are two-valued where a dead agent  knows everything. 
Different choices in designing a semantics produce  different three-valued semantics, and  also different two-valued semantics. 
In this work, we categorize the available choices by discounting the bad ones, identifying the equivalent ones, and connecting the non-equivalent ones via a translation.  The main result of the paper is identifying the main relevant distinction to be the number of truth values and  bridging this difference by means of a novel embedding from three-~into two-valued semantics. 
This translation also enables us to highlight quite fundamental modeling differences underpinning various two- and three-valued approaches in this area of combinatorial topology. In particular, pure complexes can be defined as those invariant under the translation.
\end{abstract}

\section{Introduction}
\label{introduction}

This contribution is on a topic where combinatorial topology and epistemic logic meet, using a categorically motivated duality between 
simplicial complexes, a structural primitive of a topological nature, and 
Kripke models, a structural primitive omnipresent in epistemic logic, and modal logics in general. 
Such simplicial complexes are used to describe synchronous and asynchronous computation, and the focus of our interest is simplicial complexes representing that some processes are alive~(active), whereas other processes are dead~(have~crashed). 
Such simplicial complexes are called impure. 
As we will see, their properties can be described in a two-valued epistemic semantics, but also in a three-valued epistemic semantics where the third value means undefined. 
That represents that formulas cannot be evaluated, as far as the knowledge or other features of a dead agent are concerned. 
However, there are also ways to represent dead agents in a standard truth-valued (two-valued) semantics: namely, by a dead agent knowing the false proposition (and, therefore, everything~--- in Kripke semantics it is common to say that such agents are mad, or crazy). 

\emph{In this contribution we compare two-valued and three-valued epistemic semantics for impure complexes, by way of translations capturing the three-valued meaning in a two-valued way while using the same logical language. We also discuss at length the consequences for truth and validity. These novel translations allow to highlight the respective advantages of such two-~and three-valued approaches.} 

In the remainder of this introduction we survey the research area, give detailed informal examples, and summarize the results that we will achieve. 

\paragraph*{Survey of the literature}
Combinatorial topology has been used in distributed computing to model concurrency and asynchrony since~\cite{BiranMZ90,FischerLP85,luoietal:1987}, with higher-dimensional topological properties considered in~\cite{HS99,herlihyetal:2013}. The basic structure in combinatorial topology is the \emph{simplicial complex}, a downward closed  collection of subsets, called \emph{simplices}, of a set of \emph{vertices}. Geometric manipulations such as subdivision have natural combinatorial counterparts. 

Epistemic logic investigates knowledge and belief, and change of knowledge and belief, in multi-agent systems~\cite{hintikka:1962}. Knowledge change was extensively modeled in temporal epistemic logics~\cite{Alur2002,halpernmoses:1990,Pnueli77,dixonetal.handbook:2015} and more recently in dynamic epistemic logics~\cite{baltagetal:1998,hvdetal.del:2007,moss.handbook:2015}, including synchronous~\cite{jfaketal.JPL:2009} and asynchronous~\cite{degremontetal:2011,BalbianiDG22} interpretations of dynamics. The basic structure in epistemic semantics is the \emph{Kripke model}, defined by a set of \emph{states} (or~\emph{worlds}), a collection of binary \emph{indistinguishability} relations between those states, used to interpret knowledge, and a collection of subsets of states, called \emph{properties}, for where propositional variables are true.

Fairly recently, an epistemic logic interpreted on simplicial complexes was proposed in~\cite{ledent:2019,GoubaultLR21,goubaultetalpostdali:2021}, including exact correspondence between simplicial complexes and Kripke models. Also, in those and other works~\cite{PflegerS18,diego:2021,hvdetal.simpl:2022}, the action models of~\cite{baltagetal:1998} are used to model distributed computing tasks and algorithms, with asynchrony treated as in~\cite{degremontetal:2011}. Action models, which are like Kripke models, also have counterparts that are like simplicial complexes~\cite{ledent:2019,hvdetal.simpl:2022}.

Even more recently, epistemic semantics for impure complexes have been proposed. In impure complexes some processes have crashed, i.e.,~are dead. This typically represents synchronous computation (with timeouts), as in asynchronous computation inactive processes can in principle become active again later~\cite{herlihyetal:2013}. Dead agents~--- and  live agents' uncertainty about whether those are dead~--- need some representation in a modal logical semantics. Choices occurring in the literature are to consider knowledge of dead agents either undefined~\cite{hvdetal.deadalive:2022} or  trivial~\cite{GoubaultLR22}. When such knowledge is undefined, it is not allowed to interpret~$K_a \phi$ if $a$~is dead. This results in a three-valued semantics~\cite{hvdetal.deadalive:2022} and an accompanying \mbox{$\mathbf{S5}$-like} modal logic~\cite{rojo:2023}. When such knowledge is trivial, we mean that $K_a \bot$~is true (agents going mad, or crazy; in epistemic logic, a standard trick coming with an empty accessibility relation~\cite{baltagetal:1998,gerbrandyetal:1997}), from which we derive that $K_a \phi$~is true for all formulas~$\phi$. This remains a two-valued semantics~\cite{GoubaultLR22},  and the accompanying logic is $\mathbf{KB4}$-like. An agent who is dead is not so unlike an agent who is incorrect, as in~\cite{abs-2106-11499,hvdetal.aiml:2022}. The approach of~\cite{GoubaultLR22} was generalized from individual knowledge to distributed knowledge, and from simplicial complexes to  (semi-)simplicial sets~\cite{goubaultetal:2023}.

\paragraph*{Informal examples}
Figure~\ref{figure.figure} displays some simplicial complexes and, for the benefit of the reader more familiar with that representation, corresponding Kripke models. In the depictions of Kripke models we assume reflexivity and symmetry of accessibility relations. The depicted simplicial complexes are for three agents. The vertices of a simplex are required to be labeled with different agents. A maximum size simplex, called facet, therefore, consists of three vertices. This is called dimension~$2$. These are the triangles in the figure. For two agents we get lines/edges, for four agents we get tetrahedra,~etc. A~facet corresponds to a state in a Kripke model. A label like~$0_a$ on a vertex represents that it is a vertex for agent~$a$ and that agent~$a$'s local state has value~$0$,~etc. We can see this as the boolean value of a local proposition where $0$~means false and $1$~means true. Together these labels determine the valuation in a corresponding Kripke model, for example in states labeled~$0_a1_b1_c$ agent~$a$'s value is~$0$, $b$'s~is~$1$, and $c$'s~is~$1$. The single triangle in Fig.~\ref{figure.figure}.iii corresponds to the singleton $\mathbf{S5}$~model below it, in~Fig.~\ref{figure.figure}.vi. With two triangles, if they only intersect in~$a$, as in Fig.~\ref{figure.figure}.i, it means that agent~$a$ cannot distinguish these states, as in Fig.~\ref{figure.figure}.iv, so that $a$~is uncertain about the value of~$b$; whereas if the triangles intersect in~$a$~and~$c$, as in Fig.~\ref{figure.figure}.ii, both~$a$~and~$c$ are uncertain about the value of~$b$, so in corresponding Fig.~\ref{figure.figure}.v the two states are indistinguishable for the two agents~$a$~and~$c$.  

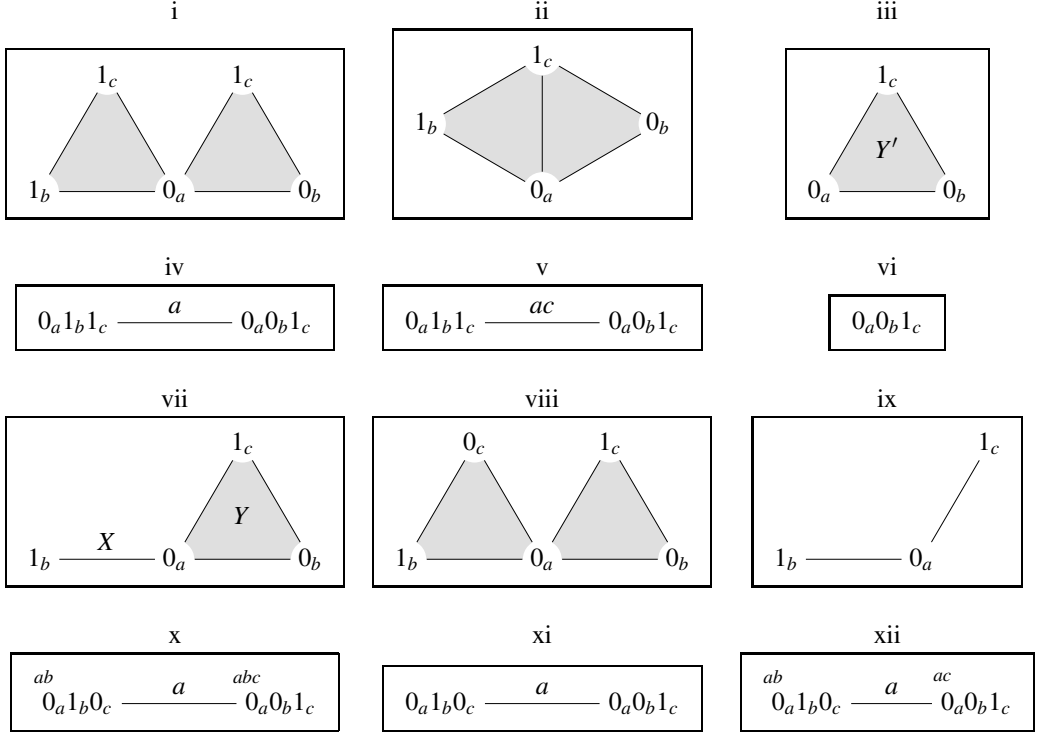
\begin{figure}[ht] \center
\scalebox{.9}{
\begin{tabular}{ccc}
i & ii & iii \\
\fbox{
\begin{tikzpicture}[round/.style={circle,fill=white,inner sep=1}]
\fill[fill=gray!25!white] (2,0) -- (4,0) -- (3,1.71) -- cycle;
\fill[fill=gray!25!white] (0,0) -- (2,0) -- (1,1.71) -- cycle;
\node[round] (b1) at (0,0) {$1_b$};
\node[round] (b0) at (4,0) {$0_b$};
\node[round] (c1) at (3,1.71) {$1_c$};
\node[round] (lc1) at (1,1.71) {$1_c$};
\node[round] (a0) at (2,0) {$0_a$};

\draw[-] (b1) -- (a0);
\draw[-] (b1) -- (lc1);
\draw[-] (a0) -- (lc1);
\draw[-] (a0) -- (b0);
\draw[-] (b0) -- (c1);
\draw[-] (a0) -- (c1);
\end{tikzpicture}
}
&
\fbox{
\begin{tikzpicture}[round/.style={circle,fill=white,inner sep=1}]
\fill[fill=gray!25!white] (2,0) -- (2,2) -- (3.71,1) -- cycle;
\fill[fill=gray!25!white] (0.29,1) -- (2,0) -- (2,2) -- cycle;
\node[round] (b1) at (.29,1) {$1_b$};
\node[round] (b0) at (3.71,1) {$0_b$};
\node[round] (c1) at (2,2) {$1_c$};
\node[round] (a0) at (2,0) {$0_a$};

\draw[-] (b1) -- (a0);
\draw[-] (b1) -- (c1);
\draw[-] (a0) -- (b0);
\draw[-] (b0) -- (c1);
\draw[-] (a0) -- (c1);
\end{tikzpicture}
}
& 
\fbox{
\begin{tikzpicture}[round/.style={circle,fill=white,inner sep=1}]
\fill[fill=gray!25!white] (2,0) -- (4,0) -- (3,1.71) -- cycle;
\node[round] (b0) at (4,0) {$0_b$};
\node[round] (c1) at (3,1.71) {$1_c$};
\node[round] (a0) at (2,0) {$0_a$};
\node (f1) at (3,.65) {$Y'$};

\draw[-] (a0) -- (b0);
\draw[-] (b0) -- (c1);
\draw[-] (a0) -- (c1);
\end{tikzpicture}
} \\ && \\
iv & v & vi \\ 
\fbox{
\begin{tikzpicture}
\node (010) at (.5,0) {$0_a1_b1_c$};
\node (001) at (3.5,0) {$0_a0_b1_c$};
\draw[-] (010) -- node[above] {$a$} (001);
\end{tikzpicture}
}
&
\fbox{
\begin{tikzpicture}
\node (010) at (.5,0) {$0_a1_b1_c$};
\node (001) at (3.5,0) {$0_a0_b1_c$};
\draw[-] (010) -- node[above] {$ac$} (001);
\end{tikzpicture}
}
&
\fbox{
\begin{tikzpicture}
\node (010) at (.5,0) {$0_a0_b1_c$};
\end{tikzpicture}
} \\ \ \\
vii & viii & ix \\
\fbox{
\begin{tikzpicture}[round/.style={circle,fill=white,inner sep=1}]
\fill[fill=gray!25!white] (2,0) -- (4,0) -- (3,1.71) -- cycle;
\node[round] (b1) at (0,0) {$1_b$};
\node[round] (b0) at (4,0) {$0_b$};
\node[round] (c1) at (3,1.71) {$1_c$};
\node[round] (a0) at (2,0) {$0_a$};
\node (f1) at (3,.65) {$Y$};

\draw[-] (b1) -- node[above] {$X$} (a0);
\draw[-] (a0) -- (b0);
\draw[-] (b0) -- (c1);
\draw[-] (a0) -- (c1);
\end{tikzpicture}
}
&
\fbox{
\begin{tikzpicture}[round/.style={circle,fill=white,inner sep=1}]
\fill[fill=gray!25!white] (2,0) -- (4,0) -- (3,1.71) -- cycle;
\fill[fill=gray!25!white] (0,0) -- (2,0) -- (1,1.71) -- cycle;
\node[round] (b1) at (0,0) {$1_b$};
\node[round] (b0) at (4,0) {$0_b$};
\node[round] (c1) at (3,1.71) {$1_c$};
\node[round] (lc1) at (1,1.71) {$0_c$};
\node[round] (a0) at (2,0) {$0_a$};

\draw[-] (b1) -- (a0);
\draw[-] (b1) -- (lc1);
\draw[-] (a0) -- (lc1);
\draw[-] (a0) -- (b0);
\draw[-] (b0) -- (c1);
\draw[-] (a0) -- (c1);
\end{tikzpicture}
}
&
\fbox{
\begin{tikzpicture}[round/.style={circle,fill=white,inner sep=1}]
\node[round] (b1) at (0,0) {$1_b$};
\node[round] (c1) at (3,1.71) {$1_c$};
\node[round] (a0) at (2,0) {$0_a$};

\draw[-] (b1) -- (a0);
\draw[-] (a0) -- (c1);
\end{tikzpicture}
} \\ \ \\
x & xi & xii \\ 
\fbox{
\begin{tikzpicture}
\node (010l) at (0,0.4) {\scriptsize$ab$};
\node (001l) at (3,0.4) {\scriptsize$abc$};
\node (010) at (.5,0) {$0_a1_b0_c$};
\node (001) at (3.5,0) {$0_a0_b1_c$};
\draw[-] (010) -- node[above] {$a$} (001);
\end{tikzpicture}
}
&
\fbox{
\begin{tikzpicture}
\node (010) at (.5,0) {$0_a1_b0_c$};
\node (001) at (3.5,0) {$0_a0_b1_c$};
\draw[-] (010) -- node[above] {$a$} (001);
\end{tikzpicture}
}
&
\fbox{
\begin{tikzpicture}
\node (010l) at (0,0.4) {\scriptsize$ab$};
\node (001l) at (2.5,0.4) {\scriptsize$ac$};
\node (010) at (.5,0) {$0_a1_b0_c$};
\node (001) at (3,0) {$0_a0_b1_c$};
\draw[-] (010) -- node[above] {$a$} (001);
\end{tikzpicture}
}
\end{tabular}
}
\caption{Examples of pure and impure simplicial complexes and corresponding Kripke models}
\label{figure.figure}
\end{figure}

The current state of the distributed system is represented by a distinguished facet of the simplicial complex, just as we need a distinguished (actual) state in a Kripke model in order to evaluate propositions. For example, in the leftmost triangle of Fig.~\ref{figure.figure}.i, as well as in the leftmost state/world of Fig.~\ref{figure.figure}.iv, $a$~is uncertain whether the value of~$b$ is~$0$~or~$1$, whereas $b$~knows that its value is~$1$, and all three agents know that the value of~$c$ is~$1$. However, any face that is not a facet may just as well be taken as the distinguished point of evaluation. For example, in the $0_a$~vertex of Fig.~\ref{figure.figure}.i it also holds that $a$~is uncertain about the value of~$c$, but $b$'s~knowledge is undefined in that vertex. The Kripke model representation in Fig.~\ref{figure.figure}.iv does not allow us such flexibility.

A complex is impure if the maximal faces do not all have the same \emph{maximum} dimension. Let us now consider some impure complexes. Fig.~\ref{figure.figure}.vii consists of two maximal facets, an edge of dimension~$1$ and a triangle of dimension~$2$. Therefore, it is impure. Fig.~\ref{figure.figure}.vii represents that agent~$a$ is uncertain whether agent~$c$ is alive, and also that agent~$a$ is uncertain about the value of agent~$b$. The latter is as in Fig.~\ref{figure.figure}.i. However, one might say that $a$~is uncertain whether Fig.~\ref{figure.figure}.vii was ``originally'' Fig.~\ref{figure.figure}.i~or~\ref{figure.figure}.viii, where $c$'s~value is~$0$ on the left. Another impure complex is Fig.~\ref{figure.figure}.ix, wherein $a$~is uncertain whether $b$~is dead or whether $c$~is dead. Although all maximal faces of this complex have the same dimension, this dimension~$1$ is smaller than~$2$, the maximum possible for three agents, and that is why this complex is impure. Below these figures we again have shown their corresponding Kripke models, of which we wish to highlight two features. The suffixes with the states indicate which agents are alive. This depiction induces a set of indistinguishability relations: for each agent, the restriction of the domain to the states where the agent is alive is an equivalence relation. This is known as a partial equivalence relation. These are symmetric and transitive relations. The other feature is that, for example, value~$0$ for~$c$ in Fig.~\ref{figure.figure}.x does not correspond to a value for~$c$ in  edge~$X$ in Fig.~\ref{figure.figure}.vii. It is a bogus value. But it does not occur in Fig.~\ref{figure.figure}.vii, which is therefore a more economical representation of the same information. Similarly, for Fig.~\ref{figure.figure}.xii versus Fig.~\ref{figure.figure}.ix. Kripke models only play a minor role in this work. Their relation to complexes is explained in depth in~\cite{hvdetal.deadalive:2022}.

Finally, we wish to point out a difference, or rather the lack thereof, between Figs.~\ref{figure.figure}.iii~and~\ref{figure.figure}.vii. In the only facet~$Y'$ of the former, agent~$a$ knows that $c$~is alive. In the latter, by contrast,  $a$~is uncertain whether $c$~is alive. In the simplicial semantics of~\cite{hvdetal.deadalive:2022} this cannot be expressed in the language, which is a feature, not a problem: these facets~$Y$~and~$Y'$ contain the same information, they make the same formulas true, and the same formulas are defined there. In this work, we will also present a richer semantics that is novel, wherein matters of life and death can be expressed in the language. This is relevant, as these semantics affect translations to two-valued semantics, and thus relate in different ways to works like~\cite{GoubaultLR22}.\looseness=-1

\paragraph*{Our results}
We propose two different logical languages and three-valued epistemic semantics for impure simplicial complexes. The primitive modalities are distributed knowledge modalities, of which individual knowledge is a special case. The languages extend each other. The extension consists of propositional variables expressing whether an agent is alive or dead. This cannot be expressed (or~defined) in the more restricted language. The extended semantics is novel. For the semantics, it may a~priori \emph{seem} to make a difference whether the point of evaluation is a facet (a~maximal face) or an arbitrary face. We show that the validities are the same either way, meaning that the logic is insensitive to this choice, and that this is the case for both languages. We then define novel translations relating the three-valued semantics to a  two-valued semantics for impure complexes. The crucial aspect is that we can translate `a~formula is defined' into a much larger formula in the same language but interpreted in a two-valued way, where `a~formula is~true' and `a~formula is~false' more obviously translate to (somewhat~shorter) two-valued correspondents. We finally discuss how our results relate to the literature and to further questions.\looseness=-1

\paragraph*{Overview}
Section~\ref{tech.preliminaries} defines the various logical languages and three-valued semantics, for which we then show some properties and give examples. Section~\ref{sec.twovalued} defines the two-valued semantics, and shows why similar properties now fail. Section~\ref{sec.translation} provides translations and proves their adequacy. Examples are also given. Section~\ref{sec.discussion} reviews our results and compares them to the literature.  

\section{Three-valued epistemic semantics for impure complexes}
\label{tech.preliminaries}

We consider a finite set~$A$ of \emph{agents} (or~\emph{processes})~$a,b,\dots$ with $|A|>1$ and a  set~$P = \bigcup_{a \in A} P_a$ of \emph{local propositional} \emph{variables} (or~\emph{local atoms}) where  $P_a$~are countable and mutually disjoint sets of \emph{local variables for agent~$a$}, denoted~$p_a, q_a, p'_a, q'_a, \dots$
We also view the agent's name~$a$ as a \emph{global propositional variable} (or \emph{global atom}) stating  that agent~$a$ is alive.

We successively define the logical languages, simplicial complexes, and related structural notions, and then give the semantics.

\begin{definition}[Languages] \label{language}
	\emph{Language~$\glocalD$ for distributed knowledge with glocal\footnote{Glocal is a portmanteau word formed from global$+$local.} atoms}   is defined by	
	\begin{equation}
	\label{eqmaingrammar}
	\phi 
	\cce 
	a \mid p_a \mid \neg\phi \mid (\phi\land\phi) \mid \N_B \phi
	\end{equation}
	where $a \in A$, $B \subseteq A$, and $p_a \in P_a$. Apart from other propositional connectives expressed  via the standard notational abbreviations, we use $D_B\phi \ce \neg \N_B \neg \phi$, $\M_a \phi \ce  \N_{\{a\}}  \phi$, and $K_a \phi \ce \neg \M_a \neg \phi$. The last two abbreviations mean that individual knowledge~$K_a$ of agent~$a$ is naturally expressed in this language as distributed knowledge~$D_{\{a\}}$ of the singleton group~$\{a\}$.
\emph{Language~$\localD$~for distributed knowledge with local atoms} is obtained from grammar~\eqref{eqmaingrammar} by dropping global atoms~$a$. 
\emph{Language~$\glocalK$ for individual knowledge with glocal atoms}  is the language where  sets~$B \subseteq A$ are restricted to singleton sets and $K_a$~and~$\M_a$~are used instead of~$D_{\{a\}}$~and~$\N_{\{a\}}$ respectively.
\emph{Language~$\localK$ for individual knowledge with local atoms} is both  restricted to individual knowledge and without global atoms.
\end{definition}

For~$D_B \phi$ we read `group~$B$ of agents have distributed knowledge of~$\phi$,' and for~$K_a\phi$ we read `agent~$a$ knows~$\phi$.' 

Next, we define our structural primitive, the simplicial model. Other than in the introduction, Kripke models play no part in this work and will not be defined.

\begin{definition}[Simplicial model]
	A \emph{simplicial model}~$\C$ is a triple~$(C,\chi,\ell)$ where $C$~is a \emph{simplicial complex}, $\chi$~is a \emph{chromatic map}, and $\ell$~is a \emph{valuation}. Here: 
	\begin{compactitem}
		\item 
		A \emph{(simplicial) complex}~$C\ne \varnothing$ is a collection of  \emph{simplices} that are non-empty finite subsets of a given set~$\VV$  of vertices such that $C$~is downward closed (i.e.,~$X \in C$ and $\varnothing \ne Y \subseteq X$ imply $Y \in C$). Simplices represent partial global states of a distributed system. 
		It is required that every vertex form a simplex by itself, i.e.,~$\bigl\{\{v\} \mid v \in \VV\bigr\} \subseteq C$.
		\item 
		Vertices represent local states of agents, with a \emph{chromatic map}~$\chi \colon \VV \to A$ assigning each vertex to one of the agents in such a way that  each agent has at most one vertex per  simplex, 
		i.e.,~$\chi(v) = \chi(u)$ for some~$v,u \in X \in C$ implies that $v = u$.
		For $X \in C$, we define $\chi(X) \ce \{\chi(v) \mid v \in X\}$ to be the set  of agents in simplex~$X$.
		\item 
		A \emph{valuation}~$\ell \colon \VV \to 2^P$ assigns to each vertex which local variables of the vertex's owner are true in it, i.e.,~$\ell(v) \subseteq P_a$ whenever $\chi(v) = a$. 
		Variables from~$P_a\setminus\ell(v)$ are false in vertex~$v$, whereas variables from~$P \setminus P_a$ do not belong to agent~$a$ and cannot be evaluated in $a$'s~vertex~$v$.
		The set of variables that are true in  simplex~$X \in C$ is given by $\ell(X)\ce\bigcup_{v \in X} \ell(v)$.
	\end{compactitem}
	If $Y \subseteq X$ for $X, Y \in C$, we say that $Y$~is a \emph{face} of~$X$. 
	Since each simplex is a face of itself, we use `simplex' and `face' interchangeably.  
	A face~$X$ is a \emph{facet} if{f} it is a maximal simplex in~$C$, i.e.,~\mbox{$Y \in C$} and $Y \supseteq X$ imply $Y = X$. 
	Facets represent global states of the distributed system, and their set is denoted~$\FF(C)$. 
	The \emph{dimension of simplex}~$X$ is~$|X|-1$, e.g.,~vertices are of dimension~$0$,  edges are of dimension~$1$,~etc. 
	The \emph{dimension of a simplicial complex (model)} is the largest dimension of its facets. 
	A simplicial complex (model) is \emph{pure} if{f} all facets have dimension~$n$ where $|A|=n+1$, i.e.,~contain vertices for all agents. 
	Otherwise it is \emph{impure}. 
	A~\emph{pointed simplicial model} is a pair~$(\C,X)$ where $X \in C$. 
\end{definition}

Having defined the logical language and the structures, we now present the three-valued semantics. We distinguish  \emph{face-semantics} that are interpreted on arbitrary faces of simplicial models from \emph{facet-semantics} that are only interpreted on facets (maximal faces). We will later prove that the three-valued face- and  facet-semantics determine the same validities, so that the difference does not matter. However, subsequently we show that for the two-valued semantics the difference matters a great deal.

\begin{definition}[Three-valued definability and satisfaction relations]
	\label{defsatallfacesglocalD}
	The \emph{definability relation\/~$\afdef $} and \emph{satisfaction relation\/~$\afsat$} are defined by induction on~$\phi \in \glocalD$. Let $\C=(C,\chi,\ell)$~be a simplicial model and $X \in C$~a~face.
	\[	\begin{array}{lcl}
		\C,X  \afdef  a &\text{ if{f} } & X \in \FF(C); \\
		\C,X  \afdef  p_a &\text{ if{f} } & a \in \chi(X); \\
		\C,X  \afdef  \lnot \phi &\text{ if{f} }&  \C, X \afdef  \phi; \\
		\C,X  \afdef  \phi \land \psi &\text{ if{f} } & \C,X \afdef  \phi \text{ and } \C,X \afdef  \psi; \\
		\C,X  \afdef  \N_B \phi &\text{ if{f} }&  \C,Y \afdef  \phi \text{ for some } Y \in C \text{ with } B \subseteq \chi(X\cap Y) .
\\
\\
		\C,X  \afsat  a &\text{ if{f} } & \C,X\afdef a \text{ and } a \in \chi(X); \\
		\C,X  \afsat  p_a &\text{ if{f} }&    p_a \in \ell(X); \\
		\C,X  \afsat  \lnot \phi &\text{ if{f} }&  \C,X \afdef  \lnot \phi \text{ and }   \C, X \nafsat  \phi; \\
		\C,X  \afsat  \phi \land \psi &\text{ if{f} }&  \C,X \afsat  \phi \text{ and } \C,X \afsat  \psi; \\
		\C,X  \afsat  \N_B \phi &\text{ if{f} }& \C,Y \afsat  \phi \text{ for some } Y \in C  \text{ with } B \subseteq \chi(X\cap Y).
	\end{array} \]
A formula~$\phi \in \glocalD$ is \emph{valid} (and~we write~$\afsat  \phi$) if{f} for any simplicial model~$\C=(C,\chi,\ell)$ and face~\mbox{$X \in C$}, we have that $\C,X \afdef  \phi$ implies $\C,X \afsat \phi$.
\end{definition}

The face-semantics for the other three languages can be derived from Definition~\ref{defsatallfacesglocalD} by restricting the formulas interpreted correspondingly (see Definition~\ref{language}).\footnote{Alternatively to declaring global atoms~$a$ definable in facets only, we could have made $a$~definable in facets and in any face where $a$~is alive. This would not have affected any results we report. We, therefore, chose the conceptually ``cleaner'' option of global variables defined only in global states.}

The semantics for the dual, box-like modality~$D_B$ can be derived from the above and is slightly more complex and less intuitive in this three-valued setting. This is why we use the diamond-like~$\N_B$ as a primitive. Note that, as any~$\phi$~is definable if{f} $\neg\phi$~is definable,  $\C,Y \afdef D_B \phi$  if{f} $\C,Y \afdef \N_B \phi$.
	\[	\begin{array}{lcl}
		\C,X  \afsat  D_B \phi &\text{ if{f} }&  \C,Y \afdef \N_B \phi \text{ and }  
(\C,Y \bowtie \phi \Rightarrow \C,Y \afsat  \phi) \text{ for all } Y \in C \text{ with } B \subseteq \chi(X\cap Y) .
	\end{array} \]

The facet-semantics can be derived from Definition~\ref{defsatallfacesglocalD} by evaluating formulas on facets only. We denote such three-valued facet-semantics by writing~$\fctdef$~and~$\fctsat$ instead of~$\afdef$~and~$\afsat$. All clauses are the same, except for~$X \in \FF(C)$,
\begin{align}
	\C,X  \fctdef  \N_B \phi &\quad\text{if{f}}\quad  \C,Y \fctdef  \phi \text{ for some } Y \in \FF(C) \text{ with } B \subseteq \chi(X\cap Y) ; \\
	\label{eq:facetsat}
	\C,X  \fctsat  \N_B \phi &\quad\text{if{f}}\quad \C,Y {\fctsat}  \phi \text{ for some } Y \in \FF(C)  \text{ with } B \subseteq \chi(X\cap Y).
\end{align}
Validity in the facet-semantics is defined as follows: formula~$\phi \in \glocalD$ is \emph{valid}, written~$\fctsat\phi$, if{f} for any simplicial model~$\C=(C,\chi,\ell)$ and facet~$X \in \FF(C)$, we have $\C,X \fctdef  \phi$ implies $\C,X \fctsat \phi$.

Various results for language~$\localK$ reported in~\cite{hvdetal.deadalive:2022} directly generalize to~$\localD$, $\glocalD$,~and~$\glocalK$. We, therefore, give those without proof: the semantics of~$\imp$, $\eq$,~and~$\lor$ are not standardly boolean but require that both arguments be defined; truth implies definability; and duality is valid:
\begin{lemma}[\cite{hvdetal.deadalive:2022}]
\begin{compactenum}
\item $\C, X \afsat \phi \lor \psi$ \quad if{f} \quad $\C, X \bowtie \phi$ and $\C, X \bowtie \psi$ and ($\C, X \afsat \phi$ or $\C, X 
 \afsat \psi$);
\item $\C, X \afsat \phi \imp \psi$ \quad if{f} \quad $\C, X \bowtie \phi$ and $\C, X \bowtie \psi$ and ($\C, X \afsat \phi \Rightarrow \C, X \afsat \psi$);
\item $\C, X \afsat \phi \eq \psi$ \quad if{f} \quad $\C, X \bowtie \phi$ and $\C, X \bowtie \psi$ and ($\C, X \afsat \phi$ if{f} $\C, X \afsat \psi$);
\item $\C, X \afsat \phi\quad \Rightarrow \quad \C, X \bowtie \phi$;
\item $\afsat D_B\phi\eq\neg\N_B\neg\phi$.
\end{compactenum}
The same properties hold for\/~$\fctdef$~and\/~$\fctsat$.
\end{lemma}  

\begin{example} \label{example.xxx}
Consider the following simplicial models~$\C$, $\C'$,~and~$\C''$ for  three agents~$a$, $b$,~and~$c$ with their values represented by local variables~$p_a$, $p_b$,~and~$p_c$  respectively. 
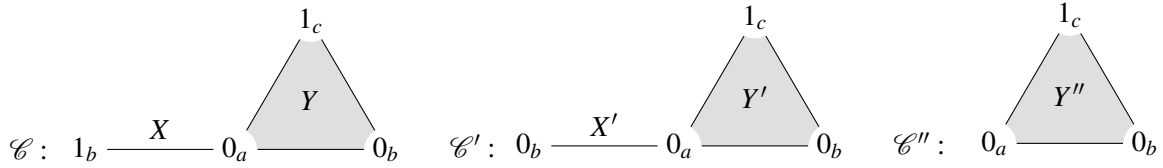
\begin{center}
\begin{tikzpicture}[round/.style={circle,fill=white,inner sep=1}]
\fill[fill=gray!25!white] (2,0) -- (4,0) -- (3,1.71) -- cycle;
\node[round] (b2) at (-0.75,0) {$\C:$};
\node[round] (b1) at (0,0) {$1_b$};
\node[round] (b0) at (4,0) {$0_b$};
\node[round] (c1) at (3,1.71) {$1_c$};
\node[round] (a0) at (2,0) {$0_a$};
\node (f1) at (3,.65) {$Y$};
\draw[-] (b1) -- node[above] {$X$} (a0);
\draw[-] (a0) -- (b0);
\draw[-] (b0) -- (c1);
\draw[-] (a0) -- (c1);
\end{tikzpicture}
\quad
\begin{tikzpicture}[round/.style={circle,fill=white,inner sep=1}]
	\fill[fill=gray!25!white] (2,0) -- (4,0) -- (3,1.71) -- cycle;
	\node[round] (b2) at (-0.75,0) {$\C':$};
	\node[round] (b1) at (0,0) {$0_b$};
	\node[round] (b0) at (4,0) {$0_b$};
	\node[round] (c1) at (3,1.71) {$1_c$};
	\node[round] (a0) at (2,0) {$0_a$};
	\node (f1) at (3,.65) {$Y'$};
	\draw[-] (b1) -- node[above] {$X'$} (a0);
	\draw[-] (a0) -- (b0);
	\draw[-] (b0) -- (c1);
	\draw[-] (a0) -- (c1);
\end{tikzpicture}
\quad
\begin{tikzpicture}[round/.style={circle,fill=white,inner sep=1}]
\fill[fill=gray!25!white] (2,0) -- (4,0) -- (3,1.71) -- cycle;
\node[round] (b1) at (1,0) {$\C'':$};
\node[round] (b0) at (4,0) {$0_b$};
\node[round] (c1) at (3,1.71) {$1_c$};
\node[round] (a0) at (2,0) {$0_a$};
\node (f1) at (3,.65) {$Y''$};
\draw[-] (a0) -- (b0);
\draw[-] (b0) -- (c1);
\draw[-] (a0) -- (c1);
\end{tikzpicture}
\end{center}
\begin{itemize}
	\item Atoms and knowledge of dead agents: Illustrating the novel aspects of the semantics, we have that $\C,X \not\bowtie p_c$  since $c \notin \chi(X)$. Consequently, $\C,X \not \bowtie \neg p_c$, $\C,X \nafsat p_c$, and $\C,X \nafsat \neg p_c$. Besides, again because $c \notin\chi(X)$, we have that $\C,X \not \bowtie \M_c \neg p_a$. Therefore, $\C,X \not \bowtie \neg \M_c \neg p_a$, $\C,X \nafsat \M_c \neg p_a$, and $\C,X \nafsat \neg \M_c \neg p_a$. 
	
	\item Knowledge of a live agent concerning dead agents: Although $\C,X\not\bowtie p_c$, after all $\C,X \afsat \M_a p_c$ because $a \in \chi(X \inter Y)$ and $\C,Y \afsat p_c$: agent~$a$ considers it possible that agent~$c$ is alive. More surprisingly, also $\C,X \afsat K_a p_c$ because given the two facets~$X$~and~$Y$ that agent~$a$ considers possible (and~all of their faces), as far as $a$~knows, $p_c$~is true. This knowledge is defeasible because $a$~may learn that the actual facet is~$X$ and not~$Y$, which she also considers possible.
	
	We further have that $\C,X\nafsat K_a p_c \imp p_c$ because $\C,X\not\bowtie p_c$ implies $\C,X\not\bowtie K_a p_c \imp p_c$. However, $\afsat K_a p_c \imp p_c$ because whenever $K_a p_c \imp p_c$ is defined,  the truth of~$K_a p_c$ implies the truth of~$p_c$.
	
	\item Individual and distributed knowledge: As expected, we have $\C,X \afsat p_b \et \neg p_a$, where the conjunct $\C,X \afsat \neg p_a$ is justified by $\C,X \bowtie p_a$  and $\C,X \nafsat p_a$. We also have $\C,Y \afsat \M_a  p_b$ because $a \in \chi(X\inter Y)$ and $\C,X \afsat p_b$. At the same time, $\C,Y \afsat D_{\{a,b\}} \neg p_b$ because all faces~$Z$ such that $\{a,b\} \subseteq \chi(Z \cap Y)$ share the $0_a$--$0_b$~edge with~$Y$ and, hence, have $\C,Z \afsat \neg p_b$.\looseness=-1
	
	Considering the global atom~$c$, we have that $\C,X \nafsat K_a \neg c$ because $a \in \chi(X \cap Y)$ and $\C,Y \nafsat \neg c$. However, $\C,X \afsat D_{\{a,b\}} \neg c$. In fact, $\C, X \afsat \phi$ if{f} $\C, X \afsat D_{\{a,b\}} \phi$: $a$~and~$b$~have distributed omniscience on  edge~$X$.
	
	\item Local and glocal languages: All formulas of language~$\localD$ are undefined/true/false in~$(\C',Y')$ if{f} they are undefined/true/false in~$(\C'',Y'')$~\cite{hvdetal.deadalive:2022}. In other words, it is impossible to express in~$\localD$ that $a$~considers it possible that $c$~is dead. By contrast, in language~$\glocalD$ we have $\C',Y' \afsat\M_a \neg c$ whereas $\C'',Y'' \nafsat\M_a \neg c$. Hence,  global atoms make the language more expressive.
\end{itemize}
\end{example}

Upwards and downwards monotonicity for the face-semantics requires certain care and is, therefore, given in some detail  for the extended languages.

\begin{lemma}[Monotonicity]
	\label{lem:MonotonicityglocalD}
	For a simplicial model~$\C=(C,\chi,\ell)$, faces~$X,Y \in C$ with $X \subseteq Y$, and formula~$\phi \in \glocalD$,	\begin{compactenum}
		\item\label{monotdefinglocalD}
		$\C,X \afdef  \phi$ implies
		$\C,Y \afdef  \phi$; \hfill (upwards monotonicity of definability)
		\item\label{monot:true.glocalD}
		$\C,X \afsat  \phi$ implies
		$\C,Y \afsat  \phi$; \hfill (upwards monotonicity of truth)
		\item\label{monotbackglocalD}
		$\C,Y \afsat  \phi$ and $\C, X \afdef \phi$ imply
		$\C,X \afsat  \phi$. \hfill (downwards monotonicity of truth modulo definability)
	\end{compactenum}
\end{lemma}
\begin{proof}
	The proof is by induction on~$\phi \in \glocalD$. All cases are as in~\cite{hvdetal.deadalive:2022} except the new cases for global atoms and (dual) distributed knowledge that are shown below:
\begin{compactdesc}
\item[Case~$a$]
Whether $\C,X \afdef  a$ or $\C,X \afsat  a$ is assumed, $X$~must be a facet, so  $Y=X$~is the same facet. Consequently, $\C,X \afdef  \phi$ implies $\C,Y \afdef  \phi$, $\C,X \afsat  \phi$ implies $\C,Y \afsat  \phi$, and $\C,Y \afsat  \phi$ and $\C, X \afdef \phi$ imply
$\C,X \afsat  \phi$.

\item[Case~$\N_B \phi$]
		\begin{enumerate}
			\item Assume that $\C,X \afdef \N_B \phi$. Then $\C,Z \afdef \phi$ for some~$Z \in C$ with $B \subseteq \chi(X \cap Z)$. Since $X \subseteq Y$, we have $X \cap Z \subseteq Y \cap Z$ and, therefore, $B \subseteq \chi(Y \cap Z)$. It follows  that $\C,Y \afdef \N_B \phi$.
			\item Assume that $\C,X \afsat \N_B \phi$. Then $\C, Z \afsat \phi$ for some~$Z \in C$ with $B \subseteq \chi(X \cap Z)$. Since $X \subseteq Y$, we have  $X \cap Z \subseteq Y \cap Z$ and, therefore, $B \subseteq \chi(Y \cap Z)$. It follows that $\C,Y \afsat \N_B \phi$.
			\item Assume that $\C,Y \afsat \N_B \phi$ and $\C,X \afdef \N_B \phi$. Then $\C,Z \afsat \phi$ for some~$Z \in C$ with \mbox{$B \subseteq \chi(Y \cap Z)$}. Additionally, $B \subseteq \chi(X)$ because $\C,X \afdef \N_B \phi$. Due to $\chi$ being chromatic, any vertex~$v$ with $\chi(v) \in B$ that belongs to~$Y$ must belong to~$X$. Hence, $B \subseteq \chi(X \cap Z)$ and, therefore, \mbox{$\C,X \afsat \N_B \phi$}. \qedhere
		\end{enumerate}
\end{compactdesc}
\end{proof}
\begin{remark}
Note that this lemma becomes trivial for~$\fctdef$~and~$\fctsat$ since $X \subseteq Y$ for facets implies $X=Y$.
\end{remark}
An axiomatization of validities in~$\localK$ is reported in~\cite{rojo:2023}, but none so far for~$\glocalK$, which is novel in this contribution. Axiomatizations of the distributed knowledge versions~$\localD$~and~$\glocalD$ should extend the axiomatizations of distributed knowledge~\cite{faginetal:1995,Roelofsen07,handbookintro:2015} and will likely be related to~\cite{goubaultetal:2023}.

\paragraph*{Comparing the three-valued face- to facet-semantics} We continue by showing that the set of validities  is the same for the three-valued face-~and  facet-semantics. This (novel) result did not initially seem obvious to us, and it plays an important role when embedding the three-valued  into two-valued semantics, as the face-semantics is infelicitous in the latter. We show the results for~$\glocalD$. The results for sublanguages~$\localD$, $\glocalK$,~and~$\localK$ follow directly.

\begin{lemma} \label{lemmai}
	\label{def.facetOnly.equiv.def.allFaces.glocalD}
	$\C,X \fctdef \phi\quad \Eq\quad\C,X \afdef  \phi$ \qquad for any~$\C = (C,\chi,\ell)$, $X \in \FF(C)$, and $\phi \in \glocalD$. 
\end{lemma}
\begin{proof}
	The proof is by induction on~$\phi \in \glocalD$. The two semantics coincide for variables and propositional connectives (in~particular, global variables are defined on all facets and only on them according to~both).  The only non-trivial case to consider is~$\N_B \phi$.
\begin{compactdesc}
\item[$\Pmi$] Assume that $\C,X \afdef \N_B \phi$. Then $\C,Y \afdef \phi$ for some face~$Y \in C$ with $B \subseteq \chi(X \cap Y)$. But $Y \subseteq Z$ for some facet~$Z \in \FF(C)$ and we have $X \cap Y \subseteq X \cap Z$. Hence, $B \subseteq \chi(X \cap Z)$ and $\C, Z \afdef \phi$ by Lemma~\ref{lem:MonotonicityglocalD}.\ref{monotdefinglocalD}. It follows from~IH that $\C,Z \fctdef \phi$. Thus, $\C,X \fctdef \N_B \phi$. 
\item[$\Imp$] Assume  $\C,X \fctdef \N_B \phi$. This direction is even simpler as here facet~$Y\in \FF(C)$ is itself a face.\qedhere
\end{compactdesc}
\end{proof}

\begin{lemma} \label{lemmaii}
	\label{sat.facetOnly.equiv.sat.allFaces.glocalD}
	$\C,X \fctsat \phi \quad\Eq\quad\C,X \afsat  \phi$ \qquad for any~$\C = (C,\chi,\ell)$, $X \in \FF(C)$, and $\phi \in \glocalD$.
	\end{lemma}
\begin{proof}
	The proof is by induction on~$\phi \in \localD$. Again the two semantics coincide  except for~$\N_B \phi$.
\begin{compactdesc}
\item[$\Pmi$]  Assume that $\C,X \afsat \N_B \phi$. Then $\C,Y \afsat \phi$ for some face~$Y \in C$ with $B \subseteq \chi(X \cap Y)$. But $Y \subseteq Z$ for some facet~$Z \in \FF(C)$ and we have $X \cap Y \subseteq X \cap Z$. Hence, $B \subseteq \chi(X \cap Z)$ and $\C, Z \afsat \phi$ by Lemma~\ref{lem:MonotonicityglocalD}.\ref{monot:true.glocalD}. We can conclude that $\C,X \fctsat \N_B \phi$. 
\item[$\Imp$] Again this case is similar but simpler as every facet~$Y\in \FF(C)$ is  a face.
	\qedhere
\end{compactdesc}	
\end{proof}

\begin{theorem} \label{theoremiii}
	\label{validity.allFaces.equiv.validity.facets.glocalD}
	$\fctsat \phi\quad\Eq\quad\afsat  \phi$ \qquad for any~$\phi \in \glocalD$. 
	\end{theorem}
\begin{proof}
\begin{compactdesc}
\item[$\Imp$] Assume that~$\fctsat \phi$ and consider an arbitrary simplicial model~$\C=(C,\chi,\ell)$ and an arbitrary face~$X \in C$ with $\C,X \afdef \phi$. Then, $X \subseteq Y$ for some facet~$Y \in \FF(C)$. We have $\C,Y \afdef \phi$ by Lemma~\ref{lem:MonotonicityglocalD}.\ref{monotdefinglocalD} and $\C,Y \fctdef \phi$ by Lemma~\ref{def.facetOnly.equiv.def.allFaces.glocalD}. Since~$\fctsat \phi$, we have $\C,Y \fctsat \phi$. It follows from Lemma~\ref{sat.facetOnly.equiv.sat.allFaces.glocalD} that $\C,Y \afsat \phi$. Hence, $\C,X \afsat \phi$ by Lemma~\ref{lem:MonotonicityglocalD}.\ref{monotbackglocalD}. Thus,~$\afsat \phi$.
\item[$\Pmi$]  Assume that~$\afsat \phi$. Then $\C,X \afdef \phi$ implies $\C,X \afsat \phi$ for any face~$X \in C$ of any simplicial model $\C=(C,\chi,\ell)$. In particular, this is the case for all facets of any simplicial complex~$\C$. We can conclude from Lemmas~\ref{def.facetOnly.equiv.def.allFaces.glocalD}~and~\ref{sat.facetOnly.equiv.sat.allFaces.glocalD} that~$\fctsat \phi$.\qedhere
\end{compactdesc}
\end{proof}

\begin{remark}
\label{remtopthreevalues}
Given that the restriction to facets does not affect  the logic of the three-valued semantics, it is worth noting that boolean constants~$\top$~and~$\bot$ are expressible in~$\glocalK$~and~$\glocalD$ in the facet-~but not in the face-semantics. Indeed, there is no formula defined in all faces of all simplicial models. Hence, no formula can be always true as~$\top$ or always false as~$\bot$. By contrast, $a \lor \neg a$~can serve as~$\top$ and $a \land \neg a$~as~$\bot$ for the facet-semantics. Languages~$\localK$~and~$\localD$, on the other hand, cannot express boolean constants in any three-valued semantics.
\end{remark}

We have shown that three-valued semantics is robust with respect to definability in that the truth value of a formula does not depend on which of the agents are crashed as long as the formula is defined. In addition, the monotonicity of definability makes it possible to restrict attention to facets only, in line with the understanding that only they represent actual global states of the distributed system. The stability of the three-valued semantics modulo the choice of a partial global state or the restriction to global states only is, in our view, a strong argument in its favor.

\section{Two-valued epistemic semantics for impure complexes} \label{sec.twovalued}

We now present a two-valued semantics for impure complexes. It is inspired by that in~\cite{GoubaultLR22} (to~which we will compare it in the final Sect.~\ref{sec.discussion}), but in this work its role is rather that of a technical tool to enable us to embed three-valued semantics, and to explain why choices essential in the three-valued setting are infelicitous or non-existent in the two-valued one.

Without the third truth value ``undefined,'' definability plays no role:  every formula is defined in every face. We, therefore,  need to define only the satisfaction relation. The languages are the same. Further simplifying the two-valued setting, we will show that the global propositional variables of~$\glocalD$ are expressible in the restricted language~$\localD$, which therefore suffices.

To distinguish the two-valued  from three-valued semantics we write~$\Vdash$ for the former to contrast it with~$\afsat$ that we used for the latter. An astute reader might notice that we use the notation~$\Vdash$ with two vertical lines for the two-valued semantics and~$\afsat$ with three vertical lines for the three-valued one to make the distinction obvious.

\begin{definition}[Two-valued facet satisfaction relation]
	\label{face-semantics.two.valued}
	We define the satisfaction relation~$\twoVsat$ by induction on~$\phi\in \glocalD$. Let $\C = (C,\chi,\ell)$ and $X \in \FF(C)$.
	\[\begin{array}{lcl}
		\C,X \twoVsat a &  \text{if{f}} & a \in \chi(X) \\
		\C, X \twoVsat p_a & \text{if{f}} & p_a \in \ell(X) \\
		\C, X \twoVsat \phi\et\psi & \text{if{f}} & \C, X \twoVsat \phi \ \text{and} \ \C, X \twoVsat \psi \\
		\C, X \twoVsat \neg \phi & \text{if{f}} & \C, X \ntwoVsat \phi \\
		\C,X\twoVsat \N_B\phi & \text{if{f}} & \C,Y \twoVsat \phi \ \text{for some } \ Y \in \FF(C) \ \text{with} \ B \subseteq \chi(X \inter Y)
	\end{array}\]
	
\end{definition}
As the superscript~$\FF$ suggests,
this is a semantics for  facets. The reason we give it as the primary in the two-valued case rather than considering alongside the semantics~$\Vdash$ for arbitrary faces, as we did for three values, is that the latter  is infelicitous, as we show in Prop.~\ref{theoremiiii}.

\begin{remark}
For all four languages, $\top \ce p_a \lor \neg p_a$ and $\bot \ce p_a \land \neg p_a$ can serve as boolean constants in the two-valued semantics.
\end{remark}

In contrast to the three-valued semantics~$\fctsat$, the semantics of, for example, implication is now the standard boolean semantics so that $\C,X \twoVsat \phi \imp \psi$ if{f} $\C,X \twoVsat \phi$ implies $\C,X \twoVsat \psi$. It is simply the version without definability requirements. Similarly, for other propositional connectives and for the dual distributed knowledge modality:	
\[
		\C,X \twoVsat D_B \phi \quad \text{if{f}} \quad \C,Y \twoVsat \phi \ \text{for all} \ Y \in \FF(C) \ \text{with} \ B \subseteq \chi(X \inter Y).
	\]

One of the important resulting differences is that, in contrast to the three-valued semantics~$\afsat$ (see Lemma~\ref{lem:MonotonicityglocalD}), the two-valued face-semantics~$\Vdash$ would not  enjoy monotonicity. Take, for  example,  model~$\C'$ from Fig.~\ref{figure.figure}.iii reprinted here,  where we give name~$X'$ to the edge~$0_a$--$0_b$. 

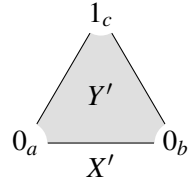
\begin{center}\begin{tikzpicture}[round/.style={circle,fill=white,inner sep=1}]
\fill[fill=gray!25!white] (2,0) -- (4,0) -- (3,1.71) -- cycle;
\node[round] (b0) at (4,0) {$0_b$};
\node[round] (c1) at (3,1.71) {$1_c$};
\node[round] (a0) at (2,0) {$0_a$};
\node (f1) at (3,.65) {$Y'$};
\node (f1) at (3,-.3) {$X'$};

\draw[-] (a0) -- (b0);
\draw[-] (b0) -- (c1);
\draw[-] (a0) -- (c1);
\end{tikzpicture}\end{center}
Note that $\C',X' \Vdash \neg p_c$ whereas $\C',Y' \Vdash p_c$. This simultaneously shows the lack of upwards  and  downwards monotonicity for all four of the languages we consider.

Unsurprisingly, this infidelity  translates to the real logical differences between the two-valued facet- and face-semantics, in contrast to Theorem~\ref{theoremiii} above. 

\begin{proposition} \label{theoremiiii}
$\Vdash$-validity is different from\/ $\twoVsat$-validity.
\end{proposition}

\begin{proof}
Consider the formula~$\phi=\M_a \top \imp \M_a  \neg p_b$ where $a,b \in A$ are two distinct agents.
We show that $\ntwoVsat \phi$ whereas $\Vdash \phi$, distinguishing the two validities.

To show that $\ntwoVsat \M_a \top \imp \M_a  \neg p_b$,  consider the following model~$\C^{\text{--}}$ with the only facet being~$X$: 
\begin{center}
\begin{tikzpicture}[round/.style={circle,fill=white,inner sep=1}]
\node[round] (b2) at  (-1,0) {$\C^{\text{--}}:$};
\node[round] (b1) at (0,0) {$0_a$};
\node[round] (a0) at (2,0) {$1_b$};

\draw[-] (b1) -- node[above] {$X$} (a0);
\end{tikzpicture}
\end{center}
Since  $a \in \chi(X)$, we have $\C^{\text{--}},X \twoVsat \M_a \top$. But  $\C^{\text{--}}, X \ntwoVsat  \neg p_b$, meaning that $\C^{\text{--}}, X \ntwoVsat  \M_a \neg p_b$. Thus, overall,  $\C^{\text{--}},X \ntwoVsat \M_a \top \imp \M_a \neg p_b$.

On the other hand, $\Vdash \M_a \top \imp \M_a \neg p_b$ for the simple reason that any face containing an $a$-vertex makes $\neg p_b$~true in that vertex. Indeed, let $X$~be an arbitrary face of an arbitrary model~$\C= (C, \chi, \ell)$ such that $\C,X \Vdash \M_a\top$. Then $\chi(v)=a$ for some vertex~$v \in X$. Since $b \notin\chi(v)$, we are guaranteed that $p_b \notin \ell(v)$, which yields $\C,v \nVdash  p_b$ and $\C,v \Vdash \neg p_b$. Finally, since $a \in \chi(X \cap \{v\})$, we have $\C, X \Vdash \M_a \neg p_b$.  Since $(\C,X)$~was arbitrary, we conclude that  $\Vdash \M_a \top \imp \M_a \neg p_b$.
\end{proof}
Proposition~\ref{theoremiiii} shows that the two-valued  face-semantics is infelicitous for impure complexes. You really do not want to have $\M_a \top \imp \M_a \neg p_b$ as a theorem. Especially since it creates an asymmetry of local truth values in light of $\nVdash \M_a \top \imp \M_a  p_b$ (for~which a singleton $a$-colored node is a countermodel). 

We will therefore, from here on, for the two-valued case consider only the   facet-semantics. 

\begin{lemma}
	\label{eq:notKfalsumasalocalD}
$\twoVsat a \eq \M_a \top$.
\end{lemma}
\begin{proof}
Since all formulas are defined, it is sufficient to show that $\C, X \twoVsat a$ if{f}  $\C, X \twoVsat \M_a \top$ for any facet~$X \in \FF(C)$ of any simplicial model~$\C = (C,\chi, \ell)$. Given that $\C, X\twoVsat \top$, we have the following equivalences: $\C, X \twoVsat a$ if{f} $a \in \chi(X)$ and, further,
\begin{equation}
\label{eqagentalive}
a \in \chi(X)
\quad\Longleftrightarrow\quad
a \in \chi(X \cap X)
\quad\Longleftrightarrow\quad
a \in \chi(X \cap X) \text{ and } \C, X \twoVsat \top
\quad\Longleftrightarrow\quad
\C, X \twoVsat \M_a\top
\qedhere
\end{equation}
\end{proof}

This lemma means that in the two-valued case global atoms are expressible already in~$\localK$. Thus, for the two-valued (facet) semantics~$\twoVsat$, we can restrict ourselves to the language~$\localD$~(or~$\localK$) without the loss of expressivity.\footnote{This result is not so unlike redefining a propositional variable~$\mathit{correct}_a$, stating that agent~$a$ is correct, as $\neg H_a \bot$ in~\cite{hvdetal.aiml:2022}, where $H_a$~is the hope modality.} We will from here on only consider language~$\localD$ for the two-valued semantics.

\begin{remark}
Note that $\fctsat a \eq \M_a\top$  also  for the three-valued  semantics but there $\top$~is only expressible in the facet-semantics and in presence of global atoms (see Remark~\ref{remtopthreevalues}). Hence, replacing~$a$ with~$\M_a \top$ would not remove global atoms from the language. Indeed, as already mentioned, global atoms are not three-valued-expressible in~$\localD$ because no formula of~$\localD$ is defined in all facets.
\end{remark}

\section{Translating three-valued into two-valued semantics} \label{sec.translation} 

We provide a translation from  language~$\glocalD$ into~$\localD$. It consists of two parts. For any formula $\phi\in\glocalD$ we define by mutual recursion 
\begin{compactitem}
\item formula~$\phi^{\bowtie}\in\localD$ that determines whether $\phi$~is \emph{defined} in some given~$(\C,X)$ and 
\item formula~$\phi^\sharp\in\localD$ that determines whether a defined formula~$\phi$ is \emph{true} in that~$(\C,X)$.
\end{compactitem} 
This covers all our tracks in three-valued semantics, as there $\phi$~may be 
\begin{compactitem}
\item undefined, in which case $\phi^{\bowtie}$~is false and, as we will see, so is~$\phi^\sharp$;
\item true, in which case both~$\phi^{\bowtie}$~and~$\phi^\sharp$ must be true; 
\item  false, in which case  $\phi^{\bowtie}$~is true but $\phi^\sharp$~is false.
\end{compactitem}

The translation from~$\glocalD$ to~$\localD$ also determines one from~$\localD$ to~$\localD$, by removing the $a^{\bowtie}$~and~$a^\sharp$~clauses.

\begin{definition}[Translations] 
	\[\begin{array}{lcl}
		a^{\bowtie} & \ce  & \top \\
		p_a^{\bowtie} & \ce  & \M_a\top \\
		(\neg\phi)^{\bowtie} & \ce  & \phi^{\bowtie} \\
		(\phi\et\psi)^{\bowtie} & \ce  & \phi^{\bowtie}\et\psi^{\bowtie} \\
		(\N_B \phi)^{\bowtie} & \ce  & \N_B \phi^{\bowtie}
	\end{array}	
\qquad\qquad
\begin{array}{lcl}
		a^\sharp & \ce  & \M_a\top \\
		p_a^\sharp & \ce  &p_a \\
		(\neg \phi)^\sharp & \ce  & (\neg\phi)^{\bowtie} \et \neg \phi^\sharp \\
		(\phi\et\psi)^\sharp & \ce  & \phi^\sharp \et \psi^\sharp \\
		(\N_B\phi)^\sharp & \ce  &  \N_B  \phi^\sharp
	\end{array}\]
\end{definition}

The main result to prove here is as follows.

\begin{theorem} \label{bowtie.glocalD}
For any 	model~$\C= (C,\chi,\ell)$, facet~$X \in \FF(C)$, and formula~$\phi \in \glocalD$
\begin{align}
\label{eqtransdef}
\C,X \fctdef \phi&\quad\Longleftrightarrow\quad\C,X \twoVsat \phi^{\bowtie};
\\
\label{eqtranssat}
\C,X \fctsat \phi &\quad \Longleftrightarrow\quad \C,X \twoVsat \phi^\sharp.
\end{align}
\end{theorem}
\begin{proof}
	We prove both statements by mutual induction on~$\phi \in \glocalD$:
\begin{compactdesc}
\item[Case~$a$] Here $a^{\bowtie}=\top$ and $a^{\sharp}=\M_a\top$. For~\eqref{eqtransdef}, both $\C,X \fctdef a$ and $\C,X \twoVsat \top$. For~\eqref{eqtranssat}, 
		\[
			\C,X \fctsat a 
			\quad \Eq \quad
			a \in \chi(X) \\
			\quad\xLeftrightarrow{\eqref{eqagentalive}}\quad
			\C, X \twoVsat \M_a\top
			\quad \Eq \quad
			\C, X \twoVsat a^\sharp.
		\]

\item[Case~$p_a$] Here $p_a^{\bowtie}=\M_a \top$ and $p_a^{\sharp}= p_a$. For~\eqref{eqtranssat}, the statement is trivial since the three-valued and two-valued definitions of satisfaction coincide for~$p_a$. For~\eqref{eqtransdef},
\[
	\C,X \fctdef p_a 
	\quad \Eq \quad
	a \in \chi(X)
	\quad\xLeftrightarrow{\eqref{eqagentalive}}\quad
	\C, X \twoVsat \M_a\top
	\quad \Eq \quad
	\C, X \twoVsat p_a^{\bowtie}.
\]

\item[Case~$\neg \phi$] Here $(\neg \phi)^{\bowtie} = \phi^{\bowtie}$  and $(\neg \phi)^\sharp  =  (\neg\phi)^{\bowtie} \et \neg \phi^\sharp$. For~\eqref{eqtransdef}, the statement follows by~IH\eqref{eqtransdef} since  $\C,X \fctdef \neg \phi$ if{f} $\C,X \fctdef \phi$. Using that,	for~\eqref{eqtranssat}, 	
\begin{multline*}
\C,X \fctsat \neg\phi 
			\quad\Eq\quad  
			\C,X \fctdef \neg \phi \text{ and } \C,X  \nfctsat \phi 
			 \quad\xLeftrightarrow{\text{\eqref{eqtransdef},IH\eqref{eqtranssat}}}\quad\\
			\C,X \twoVsat (\neg \phi)^{\bowtie} \text{ and } \C,X \ntwoVsat \phi^\sharp 
			 \quad\Eq\quad 
			\C,X \twoVsat (\neg \phi)^{\bowtie} \text{ and } \C,X \twoVsat \neg\phi^\sharp 
			 \quad\Eq\quad \\
			\C,X \twoVsat (\neg \phi)^{\bowtie}\et\neg\phi^\sharp 
			 \quad\Eq\quad 
			\C,X \twoVsat (\neg\phi)^\sharp.
		\end{multline*}

\item[Cases~$\phi \land \psi$~and~$\N_B \phi$] Here $(\phi \land \psi)^{\dag}=\phi^\dag \land \psi^\dag$ and $(\N_B\phi)^\dag=\N_B \phi^\dag$ for~$\dag \in \{\bowtie,\sharp\}$. Because the \mbox{$\bowtie$-~and~$\sharp$-translations} work the same way in both cases, the arguments for~\eqref{eqtransdef}~and~\eqref{eqtranssat} are analogous. We present the proof of the former only for~$\phi \land \psi$ and of the latter only for~$\N_B \phi$:
\begin{multline*}
\C,X \fctdef \phi \land \psi 
\quad\Eq\quad
\C,X \fctdef \phi \text{ and } \C,X \fctdef \psi
\quad\xLeftrightarrow{\text{IH\eqref{eqtransdef}}}\quad\\
\C,X \twoVsat \phi^{\bowtie} \text{ and } \C,X \twoVsat \psi^{\bowtie}
\quad\Eq\quad
\C,X \twoVsat \phi^{\bowtie} \land \psi^{\bowtie}
\quad\Eq\quad
\C,X \twoVsat (\phi \land \psi)^{\bowtie};
\end{multline*}
\vspace{-6ex}
\begin{multline*}
\C,X \fctsat \N_B\phi 
\quad\Eq\quad
\C,Y \fctsat \phi \text{ for some $Y\in\FF(C)$ with $B \subseteq \chi(X\cap Y)$}
\quad\xLeftrightarrow{\text{IH\eqref{eqtranssat}}}\quad\\
\C,Y \twoVsat \phi^\sharp \text{ for some $Y\in\FF(C)$ with $B \subseteq \chi(X\cap Y)$}
\quad\Eq\quad
\C,X \twoVsat \N_B\phi^{\sharp}
\quad\Eq\quad
\C,X \twoVsat (\N_B\phi)^\sharp.\qedhere
\end{multline*}
\end{compactdesc}
\end{proof}

By omitting the first clause in the inductive proof above, we can conclude that Theorem~\ref{bowtie.glocalD} also holds for the local language~$\localD$.

We can also represent the non-standard notion of three-valued validity in two-valued semantics:
\begin{corollary}
$\afsat \phi \quad \Eq \quad \fctsat \phi \quad \Eq \quad \twoVsat  \phi^{\bowtie} \to \phi^\sharp$ \qquad for any $\phi \in \glocalD$.
\end{corollary}
\begin{proof}
It follows from Def.~\ref{defsatallfacesglocalD} and Theorems~\ref{validity.allFaces.equiv.validity.facets.glocalD}~and~\ref{bowtie.glocalD}.
\end{proof}

It should also be noted that the two-~and three-valued semantics coincide on pure simplicial models. In fact, this agreement can serve as an independent objective distinction between pure and impure models:
\begin{corollary}
\label{corpure}
For any \underline{pure} simplicial model~$\C= (C,\chi,\ell)$, we have $\C,X \fctdef \phi$ and
\begin{equation}
\label{eqpure}
\C,X \fctsat \phi \quad \Leftrightarrow\quad \C,X \twoVsat \phi \qquad \text{for any facet $X \in \FF(C)$ and formula $\phi \in \glocalD$.}
\end{equation}
\end{corollary}
\begin{proof}
It is easy to show by induction on the construction of~$\phi$ that $\C,X \twoVsat \phi^{\bowtie}$. Indeed, for atoms both~$\top$~and $a = \M_a \top$ are true in every facet~$X$ of the pure complex~$\C$; similarly $B \subseteq A = \chi(X \cap X)$ ensures the modal clause.
The first statement now follows from~\eqref{eqtransdef}. In view of this, a simple induction argument shows that for pure models the $\sharp$-translation can be pushed through all connectives and eventually removed completely: $\C,X \twoVsat \phi \eq \phi^\sharp$. Thus, the second statement follows from~\eqref{eqtranssat}.
\end{proof}
In fact, \eqref{eqpure}~can be viewed as an alternative, functional definition of pure models.
\begin{theorem}
A simplicial model~$\C=(C,\chi,\ell)$ is pure\quad if{f}\quad \eqref{eqpure}~holds.
\end{theorem}
\begin{proof}
The only-if-direction is proved in Corollary~\ref{corpure}. For the if-direction, by contraposition, assume $\C$~is not pure, i.e.,~there is a facet~$X \in \FF(C)$ and agent~$a \in A$ such that $a \notin \chi(X)$. We, therefore, have $\C, X \nfctsat \neg p_a$ because $\C, X \not\fctdef \neg p_a$ while, at the same time, $\C,X \twoVsat \neg p_a$ because $p_a \notin \ell(X)$ in violation of~\eqref{eqpure}.
\end{proof}
\begin{example}
For the simplicial model~$\C$ in Fig.~\ref{figure.figure}.ix, we have $\C \fctsat K_a p_b \land K_a p_c$, while, at the same time, $\C \twoVsat \neg K_a p_b \land \neg K_a p_c$. This disagreement of the two semantics is why  the model in Fig.~\ref{figure.figure}.ix should not be considered pure, despite all its facets having the same dimension~$1$.
\end{example}

To conclude this section, we give some examples of the translation, and a number of derived propositions that might further throw some intuitive light on this translation (where we note once more that all these are also valid for the language~$\localD$).

\begin{example}
It is easy to see that $(\neg p_a)^\sharp = \M_a\top \et \neg p_a$ and $(\neg a)^\sharp = \top \et\neg \M_a\top$, which, modulo abbreviations and two-valued equivalences,  yields $\twoVsat (\neg p_a)^\sharp \eq a \et \neg p_a$ and $\twoVsat(\neg a)^\sharp \eq \neg a$.

	Consider two agents~$a,b \in A$, a simplicial model~$\C = (C,\chi, \ell)$, and its facet~$X \in \FF(C)$. 
	\[
	(K_a p_b)^{\bowtie}  = (\neg \M_a \neg p_b)^{\bowtie} = (\M_a \neg p_b)^{\bowtie} = \M_a  (\neg p_b)^{\bowtie} = \M_a p_b^{\bowtie} = \M_a \M_b \top;
	\]
	\vspace{-4.8ex}
	\begin{multline*}
		(K_a p_b)^{\sharp}  = (\neg \M_a \neg p_b)^{\sharp} 
		 = (\neg \M_a \neg p_b)^{\bowtie} \land \neg (\M_a \neg p_b)^{\sharp} 
		 = \M_a \M_b \top \land \neg \M_a (\neg p_b)^{\sharp} = \\
		  \M_a \M_b \top \land \neg \M_a \left((\neg p_b)^{\bowtie}\land \neg p_b^\sharp\right) 
		 = \M_a \M_b \top\land \neg \M_a \left(p_b^{\bowtie}\land \neg p_b\right) 
		 = \M_a \M_b \top \land \neg \M_a \left(\M_b\top\land \neg p_b\right).
	\end{multline*}
	Since $\twoVsat b \eq \M_b\top$ by 
	Lemma~\ref{eq:notKfalsumasalocalD} and $\twoVsat \M_b\top\land \neg p_b \eq \neg (\M_b \top \imp p_b)$, we conclude that 
	 \[
	 \C,X \fctsat K_a p_b
	 \quad\Longleftrightarrow\quad \C,X \twoVsat \M_a b \land K_a (b \imp p_b).
	 \]

This is what we want: for~$K_a p_b$ to be true, $a$~should consider it possible that $b$~is alive, and for all facets considered possible by~$a$ where $b$~is alive, $p_b$~should be true. In particular, for~$K_a p_b$ to be true it is not necessary that $b$~be actually alive. 
\end{example}
From this point on, we will routinely abbreviate~$\M_b\top$ as~$b$ without leaving language~$\localD$.

\begin{proposition}
	\label{prosharpimpliesbowtie}
	$\twoVsat \phi^\sharp \imp \phi^{\bowtie}$\qquad  for any~$\phi \in \glocalD$. 
	\end{proposition}
\begin{proof}
	It is sufficient to prove, by induction on the formula structure, that for all models~$\C=(C,\chi,\ell)$ and all facets~$X \in \FF(C)$, if $\C,X \twoVsat \phi^\sharp$, then  $\C,X \twoVsat \phi^{\bowtie}$. Since $a^{\bowtie}=\top$, the statement is trivial for global atoms. For local atoms, with $p_a^\sharp=p_a$ and $p_a^{\bowtie}=a$,
	\[
	\C,X \twoVsat p_a 
	\quad \Eq\quad 
	p_a \in \ell(X)
	\quad\Imp\quad
	a \in \chi(X)
	\quad \xLeftrightarrow{\eqref{eqagentalive}} \quad
	\C,X \twoVsat a.
	\] 
	For~$\neg \phi$, the statement follows directly from the definition of~$\sharp$. The remaining two cases easily follow by~IH. We only show the case of~$\phi \land \psi$:
\begin{multline*}
\C,X \twoVsat (\phi \et \psi)^{\sharp}
\quad\Eq\quad
\C,X \twoVsat \phi^{\sharp} \et \psi^{\sharp}
\quad \Eq \quad
\C,X \twoVsat \phi^{\sharp} \text{ and }\C,X \twoVsat \psi^{\sharp}
\quad \xRightarrow{\text{IH}} \quad\\
\C,X \twoVsat \phi^{\bowtie} \text{ and } \C,X \twoVsat \psi^{\bowtie}
\quad \Eq \quad
\C,X \twoVsat \phi^{\bowtie} \et \psi^{\bowtie}
\quad \Eq \quad
\C,X \twoVsat (\phi \et \psi)^{\bowtie}.\qedhere
\end{multline*}
\end{proof}

\begin{proposition}
	\label{prodoubleneg}
	$\twoVsat (\neg \neg \phi)^\sharp \eq \phi^{\sharp}$\qquad for any~$\phi \in \glocalD$.
\end{proposition}
\begin{proof}
	By definition, we have 
	\[
	(\neg \neg \phi)^\sharp = 
	(\neg\neg \phi)^{\bowtie} \et \neg (\neg \phi)^\sharp = 
	(\neg\phi)^{\bowtie} \et \neg ((\neg\phi)^{\bowtie} \et \neg \phi^\sharp) =
	\phi^{\bowtie} \et \neg(\phi^{\bowtie}\et \neg \phi^\sharp) .
		\]
	By standard propositional reasoning, $\twoVsat \bigl(\phi^{\bowtie} \et \neg(\phi^{\bowtie}\et \neg \phi^\sharp)\bigr) \eq  \bigl(\phi^{\bowtie} \et \neg\neg\phi^\sharp\bigr)$. Therefore, we have  $\twoVsat (\neg \neg \phi)^\sharp \eq \phi^{\bowtie} \et \phi^\sharp$. The desired statement now follows from Prop.~\ref{prosharpimpliesbowtie}.
\end{proof}

\section{Discussion and conclusion}
\label{sec.discussion}

In this paper, we analyzed and compared four different logical languages  for impure simplicial complexes and four semantics for them: two two-valued and two three-valued epistemic semantics. Our main findings can be summarized as follows:
\begin{compactitem}
	\item
	The two-valued face-semantics~$\Vdash$ is infelicitous. 
	\item 
	The three-valued facet-semantics~$\fctsat$ and face-semantics~$\afsat$ produce the same logic and, hence, can be used interchangeably.
	\item
	We provided a faithful embedding from the three-valued facet semantics~$\fctsat$  into the two-valued facet semantics~$\twoVsat$.
	\item 
	Global propositional variables describing whether agents are alive or dead increase the expressivity of the language in the three-valued case, but not in the two-valued one. 
	\item The two-valued facet-semantics~$\twoVsat$ and the three-valued facet-semantics~$\fctsat$ coincide on pure simplicial models. 
\end{compactitem}

By relating three-valued semantics  to two-valued semantics for impure complexes in a purely technical way, we hope we have filled a gap between publications like~\cite{GoubaultLR22,goubaultetal:2023} on the one hand and publications like~\cite{hvdetal.deadalive:2022,rojo:2023} on the other. Clearly, something different is going on here, but what is it exactly? Concerning the truth values, we provided the answer.  However, let us elaborate  on the other differences between such approaches. The most striking of them is that the impure complexes of~\cite{GoubaultLR22,goubaultetal:2023} do not have local propositional variables for the agents~(processes). Valuations do not apply to vertices. Instead, valuations apply to facets only. This is best explained by an example: 

Reconsider Fig.~\ref{figure.figure}. In the approach of~\cite{GoubaultLR22,goubaultetal:2023}, the modeler has to choose whether the value of~$p_c$ in~$X$ of Fig.~\ref{figure.figure}.vii is false or true, and, therefore, whether the ``original complex'' before process~$c$ became inactive, was Fig.~\ref{figure.figure}.i~or~viii. In the underlying contribution and in~\cite{hvdetal.deadalive:2022,rojo:2023} this choice is not made and left open. One could therefore consider  Fig.~\ref{figure.figure}.vii as some kind of quotient of Figs.~\ref{figure.figure}.i~and~viii following a crash. The choice made in~\cite{GoubaultLR22,goubaultetal:2023} is essential in order to still allow arbitrary values for processes and keep it possible that agents have positive knowledge. Their two-valued semantics for knowledge is a special case of the two-valued semantics given in Def.~\ref{face-semantics.two.valued}: $\C,X\twoVsat K_a \phi$ iff $\C,Y \twoVsat \phi$ for all $Y \in \FF(C)$ with $a \in \chi(X \inter Y)$. 
Applied to Fig.~\ref{figure.figure}.vii we can then only justify that $\C,X\twoVsat K_a p_c$ if $\C,X\twoVsat p_c$ and $\C,Y \twoVsat p_c$, in other words, if the bogus valuation of~$p_c$ in~$X$ made it true there. Otherwise, $a$~does not know the value of~$p_c$.

As shown, the two-valued face semantics is even infelicitous for \emph{pure} complexes, already for the simple reason that an atom~$p_a$ is false in a face~$X$ whenever $a$~is not a color in~$X$, but then `becomes' true if $X$~is contained in a facet~$Y$ where $p_a$~labels the $a$~vertex. This may suggest an insuperable problem but, not surprisingly, there are yet more different two-valued face semantics (that also differ from~\cite{GoubaultLR22,goubaultetal:2023}). To interpret formulas in faces that are not facets we can also use the multi-pointed semantics of~\cite[Sect.~`Local semantics for simplicial complexes']{hvdetal.simpl:2022}, wherein it is defined that $\C,\pmb{X}\Vdash \varphi$ for a set~$\pmb{X}$ of facets, if{f} $\C,X\Vdash \varphi$ for all~$X \in \pmb{X}$. In particular, now consider the set of facets containing a face~$X$. For a face~$X$ that is not a facet we then have that $\C,X \Vdash\varphi$ if{f} $\C,Y\Vdash \varphi$ for (the~set~of) all facets~$Y$ containing~$X$. (In general, we can even define for arbitrary faces that $\C,X \Vdash \phi$ if{f} $\C,\sstar(X)\Vdash \varphi$, where $\sstar(X) = \{Y \in C \mid X \subseteq Y\}$.) Consequently, in such an approach we would have that in the vertex~$0_a$ of Fig.~\ref{figure.figure}.i atom~$p_c$ is true (because it is true in both facets) whereas in the vertex also named~$0_a$ of Fig.~\ref{figure.figure}.viii atom~$p_c$ is false. Multi-pointed semantics are common fare in Kripke model settings, in particular for model checking applications and in dynamics~\cite{vEijck07}.

For further research, we wish to generalize our setting from simplicial complexes to (semi-)simplicial sets, and, correspondingly, from standard multi-agent Kripke models to Kripke models where each group~$B \subseteq A$ of agents has its own associated equivalence relation~$\sim_B$ and where the agents in~$B$ together may know more than the agents in~$B$ separately, even when merging their knowledge. In other words, we may then have that $\sim_B$~is strictly contained in~$\Inter_{b \in B} \sim_b$. In modal logic, such models seemed a rather technical tool so far, merely complicating the construction of canonical models, in works as~\cite{GoubaultLR22}. But in combinatorial topology, scenarios where a whole is more than the sum of its parts are very natural, as amply shown in~\cite{goubaultetal:2023}.

Another direction of further research would be the incorporation of dynamics such as in protocols.

\bibliographystyle{eptcs}
\bibliography{gandalf}

\end{document}